\documentclass[10pt]{article}
\usepackage{amsfonts, epsfig, amsmath, amssymb, color}

\usepackage[english]{babel}

\newcommand{\E}{\mathbf{E}}
\renewcommand{\P}{\mathbf{P}}

\renewcommand {\epsilon}{\varepsilon}

\newtheorem{prop}{Proposition}[section]

\DeclareMathSymbol{\ophi}{\mathalpha}{letters}{"1E}

\newcommand{\e}{\varepsilon}
\newcommand{\la}{\lambda}

\renewcommand{\phi}{\varphi}

\newcommand{\be}{\begin{equation}}
\newcommand{\ee}{\end{equation}}
\newcommand{\ben}{\begin{equation*}}
\newcommand{\een}{\end{equation*}}

\newcommand{\ba}{\begin{equation}\begin{aligned}}
\newcommand{\ea}{\end{aligned}\end{equation}}

\newenvironment{proof}{\par\noindent{\bf Proof:}}{\hfill$\blacksquare$\par}

\newfont{\cyrfnt}{wncyr10}
\def\J3{\cyrfnt{\rm \u{\cyrfnt I}}}
\def\j3{\cyrfnt{\rm \u{\cyrfnt i}}}

\allowdisplaybreaks[4]

\begin{document}
\title{One--dimensional space--discrete transport\\ subject to L\'evy perturbations 
}

\date{\today}


\author{Ilya Pavlyukevich \\ Institut f\"ur Mathematik,
Humboldt--Universit\"at zu Berlin, \\ Rudower Chaussee 25, 12489 Berlin, Germany \\
E--mail: pavljuke@math.hu--berlin.de\\
\\
Igor M.\ Sokolov \\ Institut f\"ur Physik, Humboldt--Universit\"at
zu Berlin, \\ Newtonstra\ss e 15, 12489 Berlin, Germany\\
E--mail: igor.sokolov@physik.hu--berlin.de}

\maketitle

\begin{abstract}
In this paper we study a one-dimensional space-discrete transport equation subject
to additive L\'evy forcing. The explicit form of the solutions
allows their analytic study. In particular we discuss the invariance of the covariance
structure of the stationary distribution for L\'evy perturbations with finite second
moment. The situation of more general L\'evy perturbations lacking
the second moment is considered as well. We moreover show that some of the properties
of the solutions are pertinent to a discrete system and are not reproduced by
its continuous analogue.
\end{abstract}

\textbf{Keywords:} transport equation; L\'evy process; L\'evy flights; Bessel function; stationary distribution.

\section{Introduction}

In what follows we discuss the system of coupled linear stochastic differential 
equations related to the one considered in Ref. \cite{MattinglySVE-07},
which after time-integration reads
\ba
\label{eq:m1}
\begin{cases}
a_0^\nu(t;s)=0,\\ 
\displaystyle
a_1^\nu(t;s)=a_1(s)+\int_s^t(-a^\nu_2(u;s)-\nu a^\nu_1(u;s)) \, du+ L_s(t),\\
\displaystyle
a_n^\nu(t;s)=a_n(s)+\int_s^t(a^\nu_{n-1}(u;s)-a^\nu_{n+1}(u;s)-\nu a^\nu_n(u;s))\, du,
\quad n\geq 2,\quad t\geq s,\\
\end{cases}
\ea
with $\nu\geq 0$ and for $-\infty<s\leq t<\infty$.
The discussion in \cite{MattinglySVE-07} was motivated by a description 
of hydrodynamic turbulence within a framework of simplified shell models, see
\cite{MattinglySVE-07a}.
The system \eqref{eq:m1} 
is a space-discrete transport equation with a Dirichlet boundary condition and 
a singular random perturbation at the position $n=1$.  The
interpretation of the equations corresponds to the situation when the energy
is pumped into the system at shell $n=1$ via the effective noise (or signal) term 
$L_s(t)$, is transported between the shells and dissipated due to the
effective viscosity $\nu\geq 0$. 

Let $L=(L(t))_{t\geq 0}$ be a one dimensional L\'evy process, i.e.\ a 
stochastically continuous process with stationary independent increments starting
at the origin and having right-continuous sample paths with left limits, and
\be
L_s(t):=L(t-s)-L(s),\quad -\infty<s\leq t,
\ee
be a time-shifted L\'evy process defined on the half-line $[s,+\infty)$.
It is well known that a L\'evy process $L$ is completely characterised by the 
Fourier transform of its marginal distributions which has the following form
due to the L\'evy--Hinchin formula:
\ba
\E e^{i\lambda L(t)}&=e^{t\Psi(\la)},\quad \la\in\mathbb{R},t\geq 0,\\
\Psi(\la)&=-\sigma^2\frac{\la^2}{2}+i\la\mu-\int_{\mathbb{R}\backslash\{0\}}
\Big(e^{i\la y}-1-i\la \frac{y}{1+y^2}\Big)\, \rho(dy).
\ea 
The function $\Psi(\cdot)$ is called the cumulant of $L$.
The first summand of $\Psi$ corresponds to a Brownian motion with variance $\sigma$, $\sigma\geq 0$,
the second summand determines the linear drift $\mu t$, $\mu\in\mathbb{R}$,
and the third summand corresponds to a pure jump L\'evy processes whose 
jump intensity and sizes are governed by the jump measure $\rho$ satisfying the integrability condition 
$\int_{\mathbb{R}\backslash\{0\}} (y^2\wedge 1)\,\rho(dy)<\infty$. 

In particular, if $\rho\equiv 0$, the L\'evy process $L$ is just a Brownian 
motion with drift. The case of the standard Brownian motion without drift, 
corresponding to the cumulant $\Psi(\la)=-\la^2/2$, 
is exactly the one discussed in Ref.~\cite{MattinglySVE-07}. 
Further, if $\rho(\mathbb{R}\backslash\{0\})=c\in(0,\infty)$, and 
$\sigma=\mu=0$, then 
$L$ is a compound Poisson process with intensity $c$ and jump sizes distributed with the 
probability law $\rho(\cdot)/c$. Finally, the case $\sigma=\mu=0$ 
and $\rho(dy)=c(\alpha)|y|^{-1-\alpha}dy$, $\alpha\in (0,2)$, 
corresponds to L\'evy flights with $\Psi(\lambda)=-|\lambda|^\alpha$.
We address the reader to the books 
\cite{Sato-99,Applebaum-04} for more information on L\'evy processes.

The goal of this paper consists in studying the long time behaviour of solutions of the 
system \eqref{eq:m1}. Our approach here differs from the one of \cite{MattinglySVE-07}
and is based on the explicit solution of \eqref{eq:m1}, which
allows a generalisation of the results of \cite{MattinglySVE-07}.
We moreover discuss some peculiar properties of the continuous limit of
\eqref{eq:m1} described by a partial differential equation. The
behaviour of the solutions of this continuous equation turns out to be
different from the one of its discrete analogue, Eq.(1).

\section{Explicit solution}

Let us first derive the explicit form of the solution of \eqref{eq:m1}.

\begin{prop}
The solution of \eqref{eq:m1} is given by
\ba
a_0^\nu(t;s)&=0,\\
a_n^\nu(t;s)&= e^{-\nu(t-s)}\sum_{m=1}^\infty 
a_m(s) \Big[J_{|n-m|}(2(t-s))+(-1)^{m-1} J_{n+m}(2(t-s))\Big]\\
&+\int_s^t H_n^\nu(t-r)\, dL_s(r),\quad n\geq 1,t\geq s,\\
H^\nu_n(r)&:=n\frac{J_n(2r)}{r}e^{-\nu r},
\ea 
$J_n(\cdot)$ being the Bessel functions of the first kind.
\end{prop}
\begin{proof}
For brevity we set $s=0$, $a_n^\nu:=a_n(0)$ and $a_n^\nu(t):=a_n^\nu(t;0)$. 
Since the system \eqref{eq:m1} is linear it is enough to solve the homogeneous 
deterministic
system for initial conditions $a^\nu_m=1$, $m\geq 1$, $a^\nu_n=0$, $n\neq m$, and
the non-homogeneous stochastic system with initial conditions $a_n^\nu=0$, $n\geq 0$.
Thus the solution of the homogeneous system is given by
\ba
a_n^\nu(t)&=e^{-\nu t}(J_{|n-m|}(2t)+(-1)^{m-1}J_{n+m}(2t)),\\
a_m^\nu(t)&=e^{-\nu t}(J_0(2t)+(-1)^{m-1}J_{2m}(2t)).
\ea
The proof is straightforward with help of the differentiation formulae 
(here and below, see \cite[Chapter 9]{AbramowitzS-84})
\ba
\label{eq:Jder}
2J'_n(t)&=J_{n-1}(t)-J_{n+1}(t),\quad n\geq 1,\\
J'_0(t)&=-J_1(t).
\ea
 The initial condition
is satisfied due to the relations $J_0(0)=1$, $J_n(0)=0$, $n\geq 1$.

To show that
$\int_0^t H^\nu_n(t-r)\,dL(r)$ solves the non-homogeneous equations we have
to be sure that the stochastic integral is well-defined for all $t\geq 0$.
First we note that the asymptotic expansion 
\ba
\label{eq:Jas0}
J_n(2r)&\approx \frac{r^n}{n!}, \quad r\to 0,\,n\geq 1,
\ea
implies that the integrand is a smooth and bounded function on $r\in[0,\infty)$ for
all $\nu\geq 0$. Consequently the stochastic integral is well-defined and we can easily 
calculate the characteristic function
\be
\E \exp\Big(i\la \int_0^t H^\nu_n(t-r)\,dL(r) \Big)=
\exp\Big(\int_0^t \Psi(\la H^\nu_n(r))\, dr \Big),
\quad \la\in\mathbb{R}.
\ee
For the next step of the proof we take use of the relations \eqref{eq:Jder} to show that for $n\geq 1$
\ba
\frac{d}{dr}H^\nu_n(r)&=- n\frac{J_n(2r)}{2r^2}e^{-\nu r}
+ n\frac{J_{n-1}(2r) -  J_{n+1}(2r) -n  J_{n}(2r)    }{2r}e^{-\nu r}\\
&=-\frac{J_{n-1}(2r)+J_{n-1}(2r)}{2s}e^{-\nu s}
+ n\frac{J_{n-1}(2r) -  J_{n+1}(2r) -\nu  J_{n}(2r)    }{2r}e^{-\nu r}\\
&=H^\nu_{n-1}(r)-H^\nu_{n+1}(r)-\nu H^\nu_n(r).
\ea
Thus we finish the proof with help of the intergation by parts. Indeed 
for $n\geq 1$ we have
\ba
\int_0^t& H^\nu_n(t-r)\,dL(r)=H^\nu_n(0)L(t)+
\int_0^t\Big[\int_0^u (H^\nu_n)'(t-r)\, dL(r)\Big]\, du\\
&=H^\nu_n(0)L(t)+\int_0^t (a^\nu_{n-1}(u)-a^\nu_{n+1}(u)-\nu a^\nu_n(u))\,du.
\ea
Taking into account the asymptotics \eqref{eq:Jas0} we note that $H^\nu_n(0)=1$ for
$n=1$ and $H^\nu_n(0)=0$ for $n\geq 2$, what finishes the proof.
\end{proof}

\section{Stationary distribution}

In this section we study the convergence of $a^\nu_n(t,s)$ to the stationary 
distribution. 

It is clear that due to the linearity of the system, the convergence to the
stationary regime depends on the asymptotic properties of the deterministic
solution corresponding to $L\equiv 0$.
In \cite{MattinglySVE-07}, the authors showed that 
for $(a_n^\nu)_{n\geq 1}\in l^2$, $\nu\geq 0$, the deterministic solution converges weakly to zero.
For bounded initial values of $a_n$, they gave an example of no convergence, 
namely the solution
\ba
&a_{2n-1}^0(t)=1,\\
&a_{2n}^0(t)=0,\quad n\geq 1,
\ea
which is a fixed point of the homogeneous system.
With help of the explicit formula for the solution, we can find another 
bounded initial condition, for which the corresponding solution weakly converges 
to zero. Indeed, for
$a_{2n-1}=0$, $a_{2n}=1$,  $n\geq 1$ we have:
\ba
&a^0_{2n-1}(t)=2\sum_{j=1}^{n-1} J_{2j-1}(2t)+ J_{2n-1}(2t),\\
&a^0_{2n}(t)=J_0(2t)+2\sum_{j=1}^{n-1} J_{2j}(2t)+  J_{2n}(2t).\\
\ea
The convergence follows from the asymptotic expansion
\ba
\label{eq:asinf}
J_n(2t)&\approx \sqrt{\frac{1}{\pi t}}\cos\Big(2t-\frac{n\pi}{2}-\frac{\pi}{4}\Big),
\quad t\to \infty.
\ea
The convergence is however not uniform over $n\geq 1$ since for any $t\geq 0$
\ba
2\sum_{j=1}^{n-1} J_{2j-1}(2t)&\to \int_0^{2t}J_0(r)\, dr
\ea
and
\ba
J_0(2t)+2\sum_{j=1}^{n-1} J_{2j}(2t)&\to 1,\quad n\to \infty.
\ea

To study the influence of the stochastic forcing, let us set $a_n^\nu(s):=0$ for 
$n\geq 0$, so that the deterministic part of the solution disappears. 
To obtain the invariant distribution, we consider the pull-back limit 
\be
\label{eq:s}
a^\nu_n=\lim_{s\to-\infty}\int_s^0 H^\nu_n(-r)\, dL_s(r)
\stackrel{d}{=}\int_0^{\infty} H^\nu_n(r)\, dL(r).
\ee
It is instructive to determine the Fourier transform of $a^\nu_n$ which can be found
as
\ba
F^\nu_n(\la)=\lim_{s\to -\infty} \E e^{i\la a^\nu_n(0,s)}
=\E \exp\Big( \int_0^\infty \Psi(\la H^\nu_n(r))\, dr  \Big),\quad \la\in\mathbb{R}.
\ea
First of all, we have to study the existence of the latter limit. 
Clearly, $a_n^\nu$ is a 
well-defined random variable if and only if for any finite $\la$
\ba
\int_0^\infty \Psi(\la H^\nu_n(r))\, dr<\infty.
\ea
In view of the asymptotics \eqref{eq:asinf} this is equivalent to the convergence
of the integral
\ba
\label{eq:c}
\int_1^\infty \Psi\Big(\frac{e^{-\nu r}\cos(2r)}{\sqrt{r}}\Big)\, dr<\infty,
\ea
which in turn depends on the asymptotics of $\Psi(r)$ as $r\to 0$.
We also notice that the condition \eqref{eq:c} does not depend on $n$, thus 
all random variables $a_n^\nu$ are either well-defined or not well-defined 
simultaneously.
In particular, for any $n,m\geq 1$ we can calculate the mutual Fourier transform
\be
\label{eq:F}
F_{n,m}^\nu(\la_n,\la_m):=\E e^{i(\la_n a^\nu_n+\la_m a^\nu_m)}
=\exp\Big(-\int_0^\infty \Psi(\la_n H^\nu_{n}(r)+\la_m H^\nu_{m}(r)   )\, dr \Big).
\ee
Clearly, similar formulae hold for any finite number of summands.

In the next section we study the invariant distributions $a^\nu_n$ in more detail
by discussing several special situations corresponding to different forms
of the cumulant $\Psi$. 

\section{Particular cases}

\subsection{Linear perturbation}
We start with the simplest deterministic perturbation $L(t)=\mu t$  
linear in time (corresponding to a constant forcing) 
with the cumulant $\Psi(\la)=i\la \mu$.
The Fourier transform of the stationary solution is found explicitly as
\ba
\label{eq:i1}
\E e^{i\la a_n^\nu}&=
\exp\Big(\int_0^\infty \Psi(\la H^\nu_n(r))\, dr\Big)\\
&=\exp\Big(i\la\mu \int_0^\infty H^\nu_n(r)\, dr\Big)\\
&=\exp\Big(i\la \mu  \int_0^\infty n\frac{J_n(2r)}{r}e^{-\nu r}\, dr\Big)\\
&=\exp\Big(i\la \mu n \Big(\frac{\nu}{2}+\sqrt{1+\frac{\nu^2}{4}}\Big)^{-n} \Big).
\ea 
This means that the stationary solution is a constant
\be
a_n^\nu=\mu n \Big(\frac{\nu}{2}+\sqrt{1+\frac{\nu^2}{4}}\Big)^{-n}
\ee
and as $\nu\to 0$ we obtain that
\be
a_n^\nu\approx \mu n\Big(1-\frac{n}{2}\nu+\dots\Big)\to a^0_n= \mu n.
\ee
On the other hand as $n\to\infty$, $a_n^\nu\to 0$ for $\nu>0$, whereas $a^0_n\to \infty$.

\subsection{Gaussian perturbation}
If the forcing is Gaussian, i.e.\ $\Psi(\la)=-\sigma^2\la^2/2$, it is clearly seen that
each $a^\nu_n$ is also a Gaussian random variable with the characteristic function
\ba
\label{eq:i2}
\E e^{i\la a^\nu_n}&=
\exp\Big(\int_0^\infty \Psi(\la H^\nu_n(r))\, dr\Big)\\
&=\exp\Big(-\la^2\frac{\sigma^2}{2} 
\int_0^\infty (H_n^\nu(r))^2\, dr\Big)\\
&=\exp\Big(-\la^2\frac{\sigma^2}{2} 
\int_0^\infty  n^2\frac{J_n^2(2r)}{r^2}e^{-2\nu r}\, dr\Big)\\
&=\exp\Big(- \frac{\la^2\sigma^2}{2\sqrt{\pi}}\frac{\Gamma(n-\frac{1}{2})}{\Gamma(n)}
\nu^{1-2n}{}_3F_2\Big(n-\frac{1}{2},n,n+\frac{1}{2};n+1,2n+1;-\frac{4}{\nu^2}\Big)\Big).
\ea
In the limit as $\nu\to 0$ we recover
\ba
\E e^{i\la a^\nu_n}\approx
\exp\Big(- \frac{\lambda^2\sigma^2}{2}\Big(\frac{8n^2}{\pi(4n^2-1)}-n\nu +\dots \Big)\Big)
\to
\E e^{i\la a^0_n}=\exp\Big(-\frac{\lambda^2\sigma^2}{\pi}\Big(1- \frac{1}{4n^2}\Big)^{-1}
\Big)
\ea
for all $n\geq 1$.
In particular,
\be
\E (a^0_n)^2 =\frac{2}{\pi} \sigma^2\Big(1- \frac{1}{4n^2}\Big)^{-1}\to 
\frac{2}{\pi} \sigma^2,\quad n\to \infty,
\ee
as it was discovered in \cite{MattinglySVE-07}. Taking into account Eq.\eqref{eq:F}, 
we determine the covariances,
\ba
\label{eq:i3}
\E a_m^\nu a_n^\nu
&=-\frac{\partial^2}{\partial \la_m\partial \la_n}F_{n,m}^\nu(\la_n,\la_m)\Big|_{\la_n=\la_m=0}
=\sigma^2\int_0^\infty H^\nu_m(r)H^\nu_n(r)\,dr\\
&=\sigma^2(2\nu)^{1-m-n}
\frac{(m+n-2)!}{(m-1)!(n-1)!}\\
&\times{}_4F_3\Big(\frac{m+n-1}{2},\frac{m+n}{2},\frac{m+n+1}{2},\frac{m+n}{2}+1;
m+1,n+1,m+n+1;-\frac{4}{\nu^2}\Big),\\
\E a_m^0a_n^0&=\sigma^2\frac{2}{\pi} \cos\Big(\pi\frac{m-n}{2}\Big)
\Big[\frac{1}{(m+n)^2-1}-\frac{1}{(m-n)^2-1}\Big],
\ea
where the latter formula consides with the result from \cite{MattinglySVE-07}.

\subsection{L\'evy perturbations with finite second moment}
Further, we note that the covariance structure of the stationary distribution 
is the same for all L\'evy forcings with the finite second moment. Indeed,
assume that $\E L(t)^2<\infty$ which is equivalent to the condition
that $\Psi(\cdot)\in C^2(\mathbb{R})$.
Differentiating the Fourier transform \eqref{eq:F} at zero, we determine
the mean value and the covariances of the stationary solutions $a_n^\nu$:
\ba
\E a_n^\nu &=\frac{d}{d\la}F_n^\nu(\la)\Big|_{\la=0}
 =-i\Psi'(0)\int_0^\infty H^\nu_n(r)\,dr,\\
\E (a_n^\nu)^2&=-\frac{d^2}{d\la^2}F_n^\nu(\la)\Big|_{\la=0}
=-\Psi''(0)\int_0^\infty (H^\nu_n(r))^2\,dr 
-\Big(\Psi'(0)\int_0^\infty H_n^\nu(r)\, dr\Big)^2  ,\\
\E a_m^\nu a_n^\nu
&=-\frac{\partial^2}{\partial \la_m\partial \la_n}F_{n,m}^\nu(\la_n,\la_m)\Big|_{\la_n=\la_m=0}\\
&=-\Psi''(0)\int_0^\infty H^\nu_m(r)H^\nu_n(r)\,dr-
(\Psi'(0))^2\int_0^\infty H^\nu_m(r)\,dr\int_0^\infty H^\nu_n(r)\,dr.
\ea
The closed form for the above integrals was already given in \eqref{eq:i1}, \eqref{eq:i2} 
and \eqref{eq:i3}. 
Thus our analysis shows that
the covariance structure of the stationary laws does not depend
on Gaussianity itself, and is the same for all L\'evy forcings with
equal first and second moments. 

\subsection{L\'evy flights}
Finally, for the L\'evy flights forcing with $\Psi(\la)=-\sigma^\alpha|\la|^\alpha$, 
$\alpha\in (0,2)$, $\sigma>0$, the stationary probability distribution 
exists for all $n\geq 1$, $\alpha\in (0,2)$, and $\nu> 0$ and for all 
$\alpha\in (\frac{2}{3},2)$ and $\nu= 0$, and is also $\alpha$--stable with the
characteristic function
\be
\E e^{i\la a^\nu_n}=
\exp\Big(-|\la|^\alpha\sigma^\alpha n^\alpha 
\int_0^\infty \frac{|J_n(2r)|^\alpha}{r^\alpha}e^{-\alpha\nu r}\, dr\Big).
\ee
In the zero viscosity case $\nu=0$, we observe a strange transition at the critical value 
$\alpha=2/3$ of the stability index. Formally it comes from the integrability condition
of the Bessel functions.
Indeed, in view of the asymptotics \eqref{eq:asinf} the limits
\be
\lim_{t\to\infty}\int_0^t \frac{|J_n(2r)|^\alpha}{r^\alpha}\, dr,
\lim_{t\to\infty}\int_1^t \frac{|\cos(t)|^\alpha}{r^{3\alpha/2}}\, dr
\quad \text{and}\quad 
\lim_{t\to\infty}\int_1^t \frac{1}{r^{3\alpha/2}}\, dr
\ee
either exist or not exist simultaneously, and the latter limit is finite if only if
$\alpha>2/3$. Since the stability index determines the weight of the tails
of the L\'evy flights, $\P(|L(t)|>u)=\mathcal{O}(u^{-\alpha})$, $u\to\infty$, we can conclude that
in the model under consideration, too heavy tails and big jumps of the random forcing 
lead to anomalous dispersion of the energy in space and do not allow to build up a 
stationary \emph{probability} distribution.
On the other hand as it will be shown in the next section, the critical value 
$\alpha=2/3$ has it origin not in the nature of the random forcing but rather in the
properties of the space-discrete deterministic equation. 

\section{Continuous limit}

In this section we discuss an interesting effect related to the 
continuous limit of Eq.\eqref{eq:m1}.
For any real function $f$ we introduce a discrete symmetric gradient at $x$
as
\be
\nabla_{\!h} f(x):=\frac{f(x+h)-f(x-h)}{2h},\quad h>0.
\ee
Denote also $a^\nu(t,n;s):=a^\nu_n(t;s)$. With this notation under condition $h=1$,
the system \eqref{eq:m1} takes the form
\ba
\label{eq:m2}
\begin{cases}
a^\nu(t,0;s)=0,\\ 
\displaystyle
a^\nu(t,h;s)=a(h;s)+\int_s^t( -2\nabla_{\!h} a^\nu(u,h;s)  -\nu h a^\nu(u,h;s)) \, du
+ \tfrac{1}{h}L_s(th),\\
\displaystyle
a^\nu(t,nh;s)=a(nh;s)+\int_s^t( -2\nabla_{\!h} a^\nu(u,nh;s)
-{\nu}{h} a^\nu(u,nh;s))\, du,
\quad n\geq 2,\\
\end{cases}
\ea
$\nabla_{\!h}$ denoting the discrete gradient w.r.t.\ the second variable of $a^\nu$.
We omit the discussion on the limiting behaviour of the deterministic
solution assuming that the initial values of $a^\nu$ are identically
zero.  According to \eqref{eq:s}, we find the stationary solution of
the perturbed system as
\be
a^\nu_h (n)=\lim_{s\to-\infty} \frac{n}{h}\int_{-s}^0 \frac{J_n(-2r)}{-r}e^{\nu h r}\, dL_s(hr) 
\stackrel{d}{=}
\frac{n}{h}\int_0^{\infty} \frac{J_n(2r)}{r}e^{-\nu h r}\, dL(hr), 
\label{eq:stin}
\ee
provided the latter stochastic integral exists. To study the existence, we use
the change of variables formula 
\be
\int_0^\infty f(r)\, dL(hr)\stackrel{d}{=}\int_0^\infty f(\tfrac{r}{h})\, dL(r)
\ee
for $h>0$ and a continuous function $f$ to obtain that
\ba
\label{eq:a}
a^\nu_h(n)&\stackrel{d}{=}\frac{n}{h}\int_0^\infty \frac{J_n(2r)}{r}e^{-\nu h r}\, dL(hr)
\stackrel{d}{=}\frac{n}{h}\int_0^\infty 
\frac{J_n(\frac{2r}{h})}{\frac{r}{h}}e^{-\nu r}\, dL(r).
\ea
It is clearly seen, that the stochastic integral on the right-hand side of 
\eqref{eq:a} exists for all $n\geq 1$ and $h>0$ if and only if the integrability condition
\eqref{eq:c} holds (recall constraints on $\alpha$ in the L\'evy flights case). 
Further, let $n$ and $h$ be such that $nh=x>0$. Then with help of \eqref{eq:a} we can
pass to the limit
\ba
\label{eq:f}
a^\nu_h(x)&\stackrel{d}{=}
\int_0^\infty \frac{x}{h}\frac{J_{\frac{x}{h}}( \frac{x}{h}\cdot\frac{2r}{x})}
{r}e^{-\nu r}\, dL(r)
\stackrel{h\to 0}{\to} \int_0^\infty \frac{\delta( \frac{2r}{x}-1)}
{r}e^{-\nu r}\, dL(r)=\frac{2}{x}e^{-\nu\frac{x}{2}}\dot L\Big(\frac{x}{2}\Big)
=:a^\nu(x).
\ea
To derive the latter limit we also used that $(\mu J_\mu(\mu(x+1)))_{\mu>0}$, is 
an approximating sequence for the Dirac 
$\delta$--function (see \cite{Lamborn-69a}), i.e.\ that for any appropriate test function $f$ we have
\be
\lim_{\mu\to\infty}\int_{0}^\infty \mu J_\mu(\mu(x+1)) f(x)\, dx = f(0).
\ee
This approximating system which appears here naturally as a generic 
property of the model,
also proved to be useful in treating relativistic plasma 
dispersion in the limit of zero magnetic field, see \cite{Lamborn-69}.
We emphasise that the limit \eqref{eq:f} is pointwise for $x>0$ and $h\to 0$,
and our argument has sense provided the condition \eqref{eq:c} holds and the
random variables $a_h^\nu(x)$ are well-defined.

As a partial differential equation analogue of \eqref{eq:m2}, let us consider the linear transport equation
\be
\label{eq:m3}
\begin{cases}
b_t^{\nu,\e}(t,x;s)+2 \nabla b^{\nu,\e}(t,x;s)+\nu b^{\nu,\e}(t,x;s)
=\dot L_s(t)\delta(x-\e),\quad t>s, x>0,\\
b^{\nu,\e}(s,x;s)=\phi(x),\quad x\geq 0, \quad \phi(0)=0, \\
b^{\nu,\e}(t,0;s)=0,\quad t\geq s,
\end{cases}
\ee
where $\nu\geq 0$, $\e>0$, $\delta$ denotes the Dirac function, and $\phi$
is the initial condition. 
Such a setting allows to
capture the random perturbation as a non-homogeneity in the linear transport equation.

Equation \eqref{eq:m3} can be easily solved explicitly.
Its general solution is given by
\be
b^{\nu,\e}(t,x;s)=\Phi(x-2t)+\int_s^t e^{-\nu(t-r)} 
\delta(2r-2t+x-\e)  \, dL_s(r)
\ee
for an arbitrary $\Phi(\cdot)$.
The first boundary condition yields:
\be
\Phi(u)=\phi(u+2s), \quad u\geq -2s.
\ee
The second boundary condition yields:
\ba
\Phi(-2t)&=-\int_s^t e^{-\nu(t-r)}\delta(2r-2t-\e)\, dL_s(r),
\quad t\geq s,\\
\Phi(u)&=-\int_s^{-\frac{u}{2}} e^{\nu(\frac{u}{2}+r)} \delta(2r +u-\e) 
 \, dL_s(r),\quad u\leq -2s.\\
\ea
Consequently,
\ba
b^{\nu,\e}(t,x;s)=\phi(x-2t+2s)+\int_s^t e^{-\nu(t-r)}\delta(2r-2t+x-\e)\, dL_s(r),
\quad s\leq t\leq s+x/2,
\ea
and
\ba
b^{\nu,\e}(t,x;s)&= \int_s^t e^{-\nu(t-r)}\delta(2r-2t+x)\, dL(r-s)-
\int_s^{t-\frac{x}{2}} e^{-\nu(t-r-\frac{x}{2})}\delta(2r-2t+x-\e)\, dL_s(r),
\quad t\geq s+x/2,
\ea
In particular, for $t=0$, fixed $x>0$ and $s\leq -x/2$ we have
\ba
b^{\nu,\e}(t,x;s)&=\int_{s}^0 e^{\nu r}\delta(2r+x-\e)\, dL_s(r)-
\int_{s}^{-\frac{x}{2}} e^{\nu(r+\frac{x}{2})}\delta(2r+x-\e)\, dL_s(r)\\
&=e^{-\nu\frac{x-\e}{2}}\dot L(-\tfrac{x-\e}{2})\stackrel{\e\to 0}{\to}
e^{-\nu\frac{x}{2}}\dot L(-\tfrac{x}{2})=:b^\nu(x).
\ea
This yields the stationary solution for the limiting equation for $\nu\geq 0$:
\be
\label{eq:b}
b^\nu(x)\stackrel{d}{=}e^{-\nu\frac{x}{2}}\dot L\Big(\frac{x}{2}\Big),\quad x> 0.
\ee
It is important to notice that in this case we have no constraints on the forcing $L$.

As we see, the stationary distributions of $a^\nu(x)$ and $b^\nu(x)$ are
different due to the pre-factor $\frac{2}{x}$. However, for $\nu>0$ the influence 
of this pre-exponential term
is not crucial at least for large values of $x$.
Another difference 
concerns the fact that the asymptotic solution $a^\nu$ can be obtained only
for a class of forcings satisfying condition \eqref{eq:c}, whereas the solution $b^\nu$ is
defined for all L\'evy drivers $L$.
The reason for these differences can be the inaccuracy of the approximation of
the discrete gradient $\nabla_{\!h}$ by
gradient operator  
which neglects higher order derivatives. Due to the presence of higher derivatives in
$\nabla_{\!h}$, oscillating solutions
of the deterministic system arise in a discrete system.  
These lead to complex interference
between the modes in the discrete model. This interference
makes $a^\nu$ to be typically smaller in amplitude than the solution of the continuous
model where no such complex interferences arise. Moreover, the interference
between many modes in the case of heavy-tailed forcings leads
to the divergences and implies non-existence of the corresponding
stationary probability distributions. Technically, the reason for this can be explained as follows.
The random perturbation comes in as a stochastic 
integral of an oscillating kernel $J_n(2r)/r$ which under certain conditions
can be approximated by the $\delta$--function. However the asymptotics of this 
kernel at infinity imposes restrictions on the space of admissible test functions,
as it happens in the case of zero viscosity and
L\'evy perturbations with heavy tails and, i.e.\ when  $\nu=0$ and
$\P(|L(t)|>u)=\mathcal{O}(u^{-\alpha})$, $u\to\infty$, $\alpha\in (0,2/3]$. 

\section{Conclusions} 

In this paper we studied a space-discrete transport equation subject to additive
L\'evy forcing. We derived an explicit solution of this equation in terms of the
Bessel functions of the first kind, and showed that the random forcing comes into 
solution via a stochastic integral of a certain oscillating kernel.
We generalised the results by \cite{MattinglySVE-07} on the covariance 
structure of the stationary distribution from the purely Gaussian to the general L\'evy
case, and in particular established that the covariances are identical for all
L\'evy forcings with equal first and second moments. In case of zero viscosity and 
non-Gaussian L\'evy flights forcing, we
found that the stationary probability distribution 
exists only for stability indices $\alpha\in (2/3,2)$. The explanation to this phenomenon
is given by comparison of a space-time limit of the space-discrete transport equation
with its partial differential equation analogue. An interesting approximating sequence
for a Dirac $\delta$--function appears as a natural part of the space-discrete system.
In particular, this approximating sequence allows only for test functions which have tails
lighter than $u^{-{2/3}}$. The discrepancy between the space-discrete and the PDE models 
comes from taking into account higher derivatives while considering a space-discrete 
gradient operator.

\section{Acknowledgements} The authors acknowledge financial support by 
DFG within
SFB 555 collaborative research project.

\bibliography{biblio-new}
\bibliographystyle{alpha}

\end{document}